\documentclass[11pt]{article}
\usepackage[margin=1in]{geometry}
\usepackage{amsmath,amssymb,amsthm,mathtools}
\usepackage[T1]{fontenc}
\usepackage{lmodern}
\usepackage[colorlinks=true,linkcolor=blue,citecolor=blue,urlcolor=blue]{hyperref}
\usepackage{microtype}
\usepackage{graphicx}

\newtheorem{theorem}{Theorem}[section]
\newtheorem{lemma}[theorem]{Lemma}
\newtheorem{proposition}[theorem]{Proposition}

\theoremstyle{definition}
\newtheorem{definition}[theorem]{Definition}

\newcommand{\eps}{\varepsilon}

\title{An Improved Bound for Testing $k$-Colorability}
\author{Bangjie Xu\\[2pt] \small Tsinghua University\\ \small \texttt{xbj24@mails.tsinghua.edu.cn}}
\date{}

\begin{document}
\maketitle

\begin{abstract}
We study the sample complexity of testing $k$-colorability in the dense graph model. A graph is $\varepsilon$-far from being $k$-colorable if at least $\varepsilon n^2$ edge deletions are required to make it $k$-colorable. We prove that a uniformly random induced subgraph on $O(k\log k/\varepsilon)$ vertices is non-$k$-colorable with probability at least $2/3$ whenever the original graph is $\varepsilon$-far from $k$-colorable. The proof combines a palette-based coloring process with a container-style encoding-and-reconstruction argument. This yields $O(k\log k/\varepsilon)$, which is no worse than the previous state of the art $O(k\log (1/\varepsilon)/\varepsilon)$ throughout the nonvacuous parameter range, and is an improvement when $\log k=o(\log(1/\varepsilon))$. For each fixed $k$, combining this upper bound with the known $\Omega(\log k/\varepsilon)$ lower bound gives $\Theta(1/\varepsilon)$ dependence on $\varepsilon$.
\end{abstract}

\section{Introduction}

\newcommand{\SC}{S}

Deciding whether a graph $G$ satisfies $\chi(G)\le k$ is NP-complete for every fixed $k\ge 3$~\cite{Karp1972,GJS1976}; under the Exponential Time Hypothesis it cannot be decided in time $2^{o(n)}$~\cite{IPZ2001}, where $n$ is the number of vertices; nevertheless, extensive researches have shown that for every fixed $k$ it can be decided in exponential but sub-$2^n$ time: for $k=3$ there is an algorithm running in time $O(1.3289^n)$~\cite{BE2005}, for $k=4$ one running in time $O(1.7159^n)$~\cite{WGJX2024}, and for every fixed $k$ there is $\eps_k>0$ such that $k$-colorability can be decided by a randomized algorithm in time $O((2-\eps_k)^n)$~\cite{Zamir2026}. In fact, for arbitrary $k$ and $n$ there are $n$-vertex graphs of chromatic number $k$ whose induced subgraphs on linear-size vertex sets are all $3$-colorable~\cite{Erdos1962}. This lack of local witnesses is what makes the global condition $\chi(G)>k$ hard to detect, and it is what makes testing $k$-colorability nontrivial. However, a body of work has shown that if there is a substantial gap between $G$ and a $k$-colorable graph (that is, a constant fraction of $G$'s edges has to be deleted before it becomes $k$-colorable), then even a random small induced subgraph of $G$ is non-$k$-colorable with constant probability, where the size of the subgraph depends only on $k$ and $\varepsilon$ and not on $|V(G)|$.

A graph $G$ on $n$ vertices is $\eps$-far from being $k$-colorable if at least $\eps n^2$ of its edges have to be removed to make it $k$-colorable, and the sample complexity $\SC_k(n,\eps)$ of $k$-colorability is the smallest $q$ such that, for every $\eps$-far $G$ with $n$ vertices, a uniformly random $q$-vertex induced subgraph of $G$ is non-$k$-colorable with probability at least $2/3$. The quantity $\SC_k(n,\eps)$ is an important parameter in property testing in the dense graph model~\cite{GGR1998,GT2003,Goldreich2017}; it is the minimum sample complexity of a canonical $\varepsilon$-tester for $k$-colorability on graphs over $n$ vertices. As we will soon see, $\SC_k(n,\eps)$ is in fact independent of $n$, so throughout the rest of the paper we write $\SC_k(\eps)$.

Bollob\'{a}s, Erd\H{o}s, Simonovits and Szemer\'{e}di~\cite{BESS1978} showed for $k=2$, and R\"{o}dl and Duke~\cite{RD1985} showed for every fixed $k$ the existence of a constant $f_k(\eps)$ such that any graph having sufficiently many vertices that is $\eps$-far from being $k$-colorable has an induced subgraph with no more than $f_k(\eps)$ vertices that is non-$k$-colorable; both proofs used Szemer\'{e}di's regularity lemma~\cite{Szemeredi1978}. Their proofs can be reorganized to show that a random sample of $C(\eps,k)$ (which is independent of $n$) vertices succeeds with constant probability. Goldreich, Goldwasser and Ron~\cite{GGR1998} abandoned regularity altogether and, by constructively extending partitions of the graph, obtained the polynomial bound $\SC_k(\eps)=\tilde{O}(k^2/\eps^3)$ for $k\ge 3$ and $\tilde{O}(1/\eps^2)$ for $k=2$. Alon and Krivelevich~\cite{AK2002} improved the bound to $\tilde{O}(k/\eps^2)$ for $k\ge 3$ and $\tilde{O}(1/\eps)$ for $k=2$ by growing a tree of colorings of the sample, and they also proved the lower bound $\SC_k(\eps)=\Omega(1/\eps)$. Sohler~\cite{Sohler2012} then obtained $\tilde{O}(k^6/\eps)$, by reducing the problem to the satisfiability of a constraint satisfaction problem.

More recently, Blais and Seth~\cite{BS2023,BS2024} obtained $\SC_k(\eps)=O(\frac{k}{\eps}\log^2 \frac{1}{\eps})$ by introducing the container method of extremal combinatorics~\cite{KW1982,BMS2015,ST2015} into property testing. Kushnir and Shapira~\cite{KushnirShapira2026} refined this analysis, removing a factor of $\log(1/\eps)$ and thereby attaining $O(\frac{k}{\eps} \log \frac{1}{\eps})$, which was the previous state of the art, and they also proved a lower bound $\SC_k(\eps)=\Omega(\log k/\eps)$, conjecturing that $\Theta(\log k/\eps)$ is the truth. In this paper, we prove an upper bound of $O(\frac{k}{\eps}\log k)$ (Theorem~\ref{thm:main} below). 

\begin{theorem}\label{thm:main}
For every integer $k\ge2$ and every $\eps>0$, \footnote{Our proof also applies when $k=1$. Thus, a formulation valid for all integers $k\ge1$ is $O\!\left(k\log(k+1)/\eps\right)$.}
\[
 S_k(\eps)=O\!\left(\frac{k}{\eps}\log k\right),
\]
with an absolute implied constant. More explicitly, choose
\[
 T:=\left\lceil\frac{4k}{\eps}\right\rceil,
 \qquad
 \alpha:=\frac{\eps}{4(k+1)},
\]
and set
\[
 q:=\max\left\{4T,\left\lceil\frac{4}{\alpha}\log(16k)\right\rceil\right\}.
\]
Then every $n$-vertex graph that is $\eps$-far from being $k$-colorable has a non-$k$-colorable induced subgraph on $q$ uniformly sampled vertices with probability at least $2/3$.
\end{theorem}

Comparing with the former upper bound of $O(\frac{k}{\eps}\log \frac{1}{\eps})$ of Kushnir and Shapira~\cite{KushnirShapira2026}, by the conclusion of Lemma~\ref{lem:range} below, the new bound is no worse than the previous one throughout the nonvacuous parameter range, since $\eps<1/(2k)$ implies $\log k<\log\frac{1}{\eps}$. Moreover, it is a strict improvement whenever $1/\eps$ is superpolynomial in $k$. Also, if $k$ is fixed, our upper bound matches the existing best lower bound of $\Omega(\frac{\log k}{\eps})$ up to a constant factor, which shows that the sample complexity (treating $k$ as constant) is $\Theta(\frac{1}{\eps})$.

In general, a container argument compresses a large family of combinatorial objects by recording a relatively small fingerprint and associating it with a container that covers every object having that fingerprint. The counting is based on the fingerprints and the containers they correspond to. In this paper, the proof idea is a combination of containers and palettes. The use of palettes and lists of available colors is classical in graph coloring, see, for example, work on bounded-palette list coloring and palette sparsification \cite{BonamyKang2017,AssadiChenKhanna2019,Ferdous2024,CzumajPengShiSohler2026}. 

Our algorithmic mechanism is different from \cite{BS2023,KushnirShapira2026}'s. For each original vertex, we use the minimum monochromatic degree over all colors that remain available in its palette.  In our process, the deletion triggered by a degree threshold empties the entire palette of an original vertex, rather than removing a single vertex--color copy. After defining a graph consisting of $k$ copies of the original graph, \cite{BS2023,KushnirShapira2026} introduce a container process that, in fact, applies to any undirected graph. In this paper, we restate the process from the perspective of palettes to make the proof appear more natural.

\section{Notation and conventions}\label{sec:notation}

We work throughout with finite simple graphs. In particular, all graphs considered below are nonempty, and we write $n:=|V(G)|\ge 1$. For an integer $k\ge 1$, write $[k]=\{1,\dots,k\}$. All logarithms are natural.

A graph $G$ is called $\eps$-far from being $k$-colorable if at least $\eps n^2$ edges must be deleted from $G$ in order to obtain a $k$-colorable graph. The sample complexity $\SC_k(\eps)$ is the smallest positive integer $q$ with the following property: for every $\eps$-far graph $G$, a uniformly random $q$-vertex induced subgraph is non-$k$-colorable with probability at least $2/3$. If $q>n$, we sample all vertices.

For a map $c:V(G)\to[k]$, let $M(c)$ denote the set of monochromatic edges, that is, edges $uv$ with $c(u)=c(v)$. Thus $|M(c)|$ is the number of monochromatic edges under $c$.

A palette system $P$ assigns to each vertex $v$ a set $P(v)\subseteq[k]$, and empty palettes are allowed. Define
\[
 U(P):=\{v:P(v)\neq\varnothing\}.
\]
For $i\in P(v)$, set
\[
 d_i^P(v):=|\{u\in N_G(v):i\in P(u)\}|,
 \qquad
 m_P(v):=\min_{i\in P(v)}d_i^P(v).
\]
The label count of $P$ is
\[
 \lambda(P):=\sum_{v\in V(G)}|P(v)|.
\]
We write $P'\subseteq P$ if $P'(v)\subseteq P(v)$ for every vertex $v$.

For positive integers $r$, let
\[
 P_N^r:=N(N-1)\cdots(N-r+1)
\]
be the number of ordered $r$-tuples of distinct elements from an $N$-element set. We set $P_N^r=0$ when $r>N$.

\section{Distance characterization}\label{sec:distance}

\begin{lemma}\label{lem:distance}
A graph $G$ is $\eps$-far from being $k$-colorable if and only if every map $c:V(G)\to[k]$ satisfies
\[
 |M(c)|\ge \eps n^2.
\]
\end{lemma}

\begin{proof}
For a fixed map $c$, deleting all monochromatic edges makes $c$ a proper $k$-coloring. Conversely, if deleting an edge set makes $c$ proper, then that edge set contains every monochromatic edge of $c$. Hence the minimum number of deletions needed to obtain a $k$-colorable graph is
\[
\min_{c:V(G)\to[k]}|M(c)|,
\]
which proves the claim.
\end{proof}

\begin{lemma}\label{lem:range}
If there exists an $n$-vertex $\eps$-far graph with $n\ge1$, then
\[
\eps<\frac1{2k}.
\]
\end{lemma}

\begin{proof}
Color every vertex independently and uniformly from $[k]$. Every edge is monochromatic with probability $1/k$, and therefore
\[
\mathbb E|M(c)|=\frac{|E(G)|}{k}\le \frac{\binom n2}{k}=\frac{n(n-1)}{2k}<\frac{n^2}{2k}.
\]
Thus some map $c_0$ satisfies $|M(c_0)|<n^2/(2k)$. By Lemma~\ref{lem:distance}, every map has at least $\eps n^2$ monochromatic edges. Hence
\[
\eps n^2\le |M(c_0)|<\frac{n^2}{2k},
\]
so $\eps<\frac{1}{2k}$.
\end{proof}

\section{A reconstruction lemma}\label{sec:reconstruction}

Our aim is to construct a $k$-coloring with relatively few monochromatic edges. We do this by coloring the vertices in several stages. When coloring a new stage, we only need to count monochromatic edges joining its vertices to vertices in the same stage or in earlier stages. The general strategy is therefore to color the vertices from easy to difficult: vertices with large palette degrees are harder to handle and are deferred to later stages. At a given stage, vertices deferred to later stages are temporarily ignored. This is reflected in the construction below by emptying the palettes of these ``more difficult'' vertices. To control the total number of monochromatic edges, we assign each vertex a time at which its color is easiest to handle. The running minimum $E_t$ records these best available thresholds.

\begin{lemma}\label{lem:reconstruction}
Let $P^0\supseteq P^1\supseteq\cdots\supseteq P^T$ be palette systems with $P^0(v)=[k]$ for every $v$. Suppose that for every $1\le t\le T$ and every $v\in U(P^t)$ there is a color $i\in P^t(v)$ such that
$d_i^{P^t}(v)\le D_t$, where $D_t\in[0,n]$, and suppose
\[
\sum_{t=1}^T D_t\le kn.\tag{1}
\]
Let
\[
 a_t:=|U(P^{t-1})\setminus U(P^t)|,
 \qquad
 B_t:=\sum_{s=1}^t a_s=n-|U(P^t)|.
\]
If
\[
 B_t\le \alpha tn\qquad(1\le t\le T),\tag{2}
\]
then there is a map $c:V(G)\to[k]$ with
\[
 |M(c)|\le\left(\alpha(k+1)+\frac{k}{T}\right)n^2.\tag{3}
\]
\end{lemma}

\begin{proof}
Let $r$ be the first index attaining $\min_{1\le t\le T}D_t$. Then $D_r\le kn/T$ since $\sum_{t=1}^T D_t\le kn$. Define
\[
 E_0:=n,
 \qquad
 E_t:=\min_{1\le s\le t}D_s.
\]
Thus $E_0\ge E_1\ge\cdots\ge E_T$ and $\sum_{t=1}^T E_t\le kn$.

Let $1=r_1<r_2<\cdots<r_j=r$ be the record indices of the running minimum, so that $r_{\ell+1}$ is the next strict decrease after $r_\ell$. Then
\[
E_{t-1}=D_{r_\ell}
\qquad\text{whenever }r_\ell<t\le r_{\ell+1}.
\tag{4}
\]
If $r=1$, the list contains only $r_1=1$ and the intermediate layers below are empty.

Partition $V(G)$ into
\[
V(G)\setminus U(P^{r_1}),
\quad
U(P^{r_\ell})\setminus U(P^{r_{\ell+1}})\ (1\le\ell<j),
\quad
U(P^r).
\]
First color $U(P^r)$, assigning to each vertex a witness color of degree at most $D_r$. Since $|U(P^r)|\le n$, counting all possible neighbors gives at most $nD_r$ monochromatic edges inside this part.

Next add the remaining layers in reverse order. For a vertex in $U(P^{r_\ell})\setminus U(P^{r_{\ell+1}})$, use a witness color in $P^{r_\ell}$ of degree at most $D_{r_\ell}$. Every already colored vertex uses a color belonging to its palette at time $r_\ell$, by nestedness. Hence this vertex creates at most $D_{r_\ell}$ new monochromatic edges to previously colored vertices. Finally color $V(G)\setminus U(P^1)$ arbitrarily, charging at most $n$ edges per vertex.

Using (4), the total number of monochromatic edges is at most
\[
 nD_r+\sum_{t=1}^r a_tE_{t-1}.\tag{5}
\]
With $B_0=0$, summation by parts gives
\[
\begin{aligned}
\sum_{t=1}^r a_tE_{t-1}
&=B_rE_{r-1}+\sum_{t=1}^{r-1}B_t(E_{t-1}-E_t)\\
&\le \alpha n\left(rE_{r-1}+\sum_{t=1}^{r-1}t(E_{t-1}-E_t)\right)\\
&=\alpha n\sum_{t=0}^{r-1}E_t\\
&=\alpha n\left(E_0+\sum_{t=1}^{r-1}E_t\right)\\
&\le \alpha n\left(n+kn\right)\\
&\le \alpha(k+1)n^2.
\end{aligned}
\]
Together with $nD_r\le kn^2/T$, we have reached the conclusion that there is a way to color the graph as required.
\end{proof}

\section{The trial coloring process}\label{sec:process}

The trial coloring process calculates palettes $Q^0, P^1, Q^1, P^2, Q^2,\ldots$ and chooses pivots $f_1,\ldots,f_T$ using colors $i_1,\ldots,i_T$. Two features of this process are crucial. First, it only removes colors from palettes and never adds any. Second, it is dynamic: each operation is based on the current palette system, rather than on the initial graph alone.

Fix a total order on $V(G)$ for tie-breaking. Given $A\subseteq V(G)$ and prescribed colors $i_1,i_2,\ldots\in[k]$, the trial coloring process is defined as follows.

\begin{quote}
\textbf{Trial coloring process.}\\
\textbf{Input:} A graph $G=(V,E)$, an integer $k\ge1$, a set $A\subseteq V$, a round limit $T$, a prescribed color sequence $(i_t)_{t\ge1}$, and a fixed total order on $V$ for tie-breaking.\\
\textbf{Output:} The pivots, thresholds, cap palettes, and palette systems generated before the process stops or completes $T$ rounds.
\begin{tabbing}
\hspace{1.5em}\=\hspace{1.5em}\=\kill
1. \> For every $v\in V$, set $Q^0(v)\gets[k]$; set $P^0\gets Q^0$ and $F_0\gets\varnothing$.\\
2. \> \textbf{for} $t=1,2,\ldots,T$ \textbf{do}\\
3. \> \> If some $v\in A$ has $Q^{t-1}(v)=\varnothing$, or $A\setminus F_{t-1}=\varnothing$, stop.\\
4. \> \> Choose $f_t\in A\setminus F_{t-1}$ maximizing $m_{Q^{t-1}}(v)$, breaking ties by the fixed order; set $D_t\gets m_{Q^{t-1}}(f_t)$.\\
5. \> \> Form the cap palette $P^t$ from $Q^{t-1}$ according to the rule below.\\
6. \> \> If $i_t\notin P^t(f_t)$, stop.\\
7. \> \> Set $F_t\gets F_{t-1}\cup\{f_t\}$ and initially set $Q^t(w)\gets P^t(w)$ for every $w\in V$.\\
8. \> \> Set $Q^t(f_t)\gets\{i_t\}$; for every $u\in N_G(f_t)$, set $Q^t(u)\gets P^t(u)\setminus\{i_t\}$.\\
9. \> \textbf{end for}\\
10. \> \textbf{return} the generated pivots, thresholds, cap palettes, and palette systems.
\end{tabbing}
The cap update in step 5 is
\[
P^t(w)=\begin{cases}
\varnothing,&Q^{t-1}(w)\ne\varnothing\text{ and }m_{Q^{t-1}}(w)>D_t,\\
Q^{t-1}(w),&\text{otherwise}.
\end{cases}
\]
\end{quote}

In particular,
\[
 U(P^t)=\{w:Q^{t-1}(w)\ne\varnothing\text{ and }m_{Q^{t-1}}(w)\le D_t\}.\tag{6}
\]

In step 5, all degrees and the threshold are computed from the palette system $Q^{t-1}$ before any cap or neighbor update.

\section{Analysis of the trial coloring process}\label{sec:process-analysis}

\subsection{Process invariants}\label{sec:invariants}

\begin{lemma}\label{lem:invariants}
For every completed run:
\begin{enumerate}
\item $P^0\supseteq P^1\supseteq\cdots$, and every $v\in U(P^t)$ has a color $i\in P^t(v)$ with $d_i^{P^t}(v)\le D_t$; moreover $0\le D_t\le n$.
\item On every run completing $T$ rounds,
\[
\sum_{t=1}^T D_t\le kn.\tag{7}
\]
\item If $G[A]$ has a proper coloring $\varphi:A\to[k]$, then, whenever round $t$ is reached, the prescribed choice $i_t=\varphi(f_t)$ is legal. If $|A|\ge T$, it completes $T$ rounds and preserves $\varphi(v)$ in the palette of every $v\in A$ throughout; in particular $A\subseteq U(P^t)$ for all $t\le T$.
\end{enumerate}
\end{lemma}

\begin{proof}
For the first assertion, caps only empty entire palettes and never add labels. If $v\in U(P^t)$, then $m_{Q^{t-1}}(v)\le D_t$ by (6). Choose $i^*\in Q^{t-1}(v)$ attaining this minimum. Since $v$ was not emptied, $P^t(v)=Q^{t-1}(v)$, and hence
\[
 d_{i^*}^{P^t}(v)\le d_{i^*}^{Q^{t-1}}(v)=m_{Q^{t-1}}(v)\le D_t.
\]
The bounds on $D_t$ are immediate from its definition and $|N_G(f_t)|\le n-1$.

For the second assertion, since $i_t\in Q^{t-1}(f_t)$,
\[
 d_{i_t}^{Q^{t-1}}(f_t)\ge D_t.
\]
Choose $D_t$ distinct neighbors $u_1,\ldots,u_{D_t}$ carrying color $i_t$ in their palettes. At the end of the round, each such vertex has lost the label $i_t$: it is either emptied by the cap or loses $i_t$ in the neighbor update. Therefore
\[
 \lambda(Q^{t-1})-\lambda(Q^t)\ge D_t.
\]
Summing over rounds and telescoping yields (7), since $\lambda(Q^0)=kn$.

For the third assertion, let $\varphi$ be a proper coloring of $G[A]$. We prove inductively that, after each round, $\varphi(v)$ remains in the palette of every $v\in A$; moreover, each selected pivot $f_s$ has the singleton palette $\{\varphi(f_s)\}$, and no neighbor of $f_s$ carries $\varphi(f_s)$ in its palette. These assertions hold initially, with the pivot assertions vacuous.

Suppose they hold before round $t$. Every vertex of $A$ is active. If $|A|\ge T$ and $t\le T$, then an unchosen vertex remains, so $f_t$ is defined. Every unchosen $v\in A$ is a candidate in the maximization defining $f_t$, and consequently
\[
 m_{Q^{t-1}}(v)\le m_{Q^{t-1}}(f_t)=D_t.
\]
Thus the cap does not empty its palette and preserves $\varphi(v)$. For an earlier pivot $f_s$, the inductive assertions give
\[
 d_{\varphi(f_s)}^{Q^{t-1}}(f_s)=0,
\]
so its singleton palette is also preserved by the cap. In particular, $\varphi(f_t)\in P^t(f_t)$, making the prescribed color choice legal.

The update deletes $\varphi(f_t)$ only from neighbors of $f_t$. Properness of $\varphi$ ensures that no such neighbor in $A$ loses its prescribed color. The pivot $f_t$ receives the singleton palette $\{\varphi(f_t)\}$, and all its neighbors lose that color. The assertions for earlier pivots remain valid. This completes the induction and proves that the run completes $T$ rounds.\end{proof}

\subsection{The missing-cap lemma}\label{sec:missing}

\begin{lemma}\label{lem:missing}
Suppose that $G$ is $\eps$-far from being $k$-colorable. Let
\[
T:=\left\lceil\frac{4k}{\eps}\right\rceil \ge 1,
\qquad
\alpha:=\frac{\eps}{4(k+1)}.
\]
Every run completing $T$ rounds has some $1\le t\le T$ with
\[
 |U(P^t)|\le (1-\alpha t)n.\tag{8}
\]
Consequently, if $|A|\ge T$ and $G[A]$ is $k$-colorable, then the proper-coloring run contains all of $A$ in the cap at that time.
\end{lemma}

\begin{proof}
By Lemma~\ref{lem:range}, whenever a relevant graph exists we have $0<\eps<1/(2k)$. Hence
\[
\alpha T\le\frac{4k+\eps}{4(k+1)}<1,
\]
so all right-hand sides in (8) are positive.

Assume for contradiction that (8) fails for every $t$. Then
\[
B_t=n-|U(P^t)|<\alpha tn,
\]
so $B_t\le\alpha tn$. Lemma~\ref{lem:invariants} supplies all the other hypotheses of Lemma~\ref{lem:reconstruction}, including (7). Therefore some coloring $c$ satisfies
\[
 |M(c)|\le\left(\alpha(k+1)+\frac{k}{T}\right)n^2
\le\frac\eps2n^2<\eps n^2,
\]
contradicting Lemma~\ref{lem:distance}. The final statement follows from the last part of Lemma~\ref{lem:invariants}.\end{proof}

\subsection{Deletion stability}\label{sec:stability}

\begin{lemma}\label{lem:stability}
Suppose that a run executed on $A$ has first $t$ pivots $f_1,\ldots,f_t$. If
\[
\{f_1,\ldots,f_t\}\subseteq A'\subseteq A,
\]
then running on $A'$ with the same successive color choices produces the same first $t$ pivots, thresholds, caps, and palette updates.
\end{lemma}

\begin{proof}
The main idea is that the choices of pivots and the cap palettes do not change. Induct on the rounds. The initial palettes are independent of the sample set. If the first $s-1$ rounds agree, the current palettes agree pointwise. The original pivot $f_s$ remains a candidate in $A'$, and every other candidate in $A'$ was also a candidate in $A$. Thus its maximality and the fixed tie-breaking order are unchanged. The cap and update depend only on the current palettes, the pivot, and the chosen color. Since the original run completed $t$ rounds, the same pointwise equality also prevents the restricted run from stopping earlier.\end{proof}

\section{Sampling upper bound}\label{sec:sampling}

The container viewpoint makes the counting strategy explicit. For each $k$-colorable induced subgraph on $q$ vertices, the process selects pivots and their colors, which determine the relevant palette container when replayed on a suitable subset of the sample. The counting proof supplements this pivot fingerprint with a block choice and the ordered vertices in the selected blocks; this replay seed contains all pivots and, by deletion stability, reproduces the same container as the full-sample run. The remaining sample positions must lie inside that container. We count these descriptions and bound the probability of such a completion, then take a union bound over the descriptions.

For each $t\in \{1,2,...,T\}$, partition the positions $[q]$ deterministically into $2t$ nonempty blocks whose sizes differ by at most one. 

Our counting argument encodes each colorable sample by one of a small number of \emph{descriptions}. We first record what such a description is.

\begin{definition}[Description]\label{def:description}
A \emph{description} is a tuple
\[
 (t,\,W,\,\mathbf i,\,u_1,\ldots,u_a,\,v_1,\ldots,v_b)
\tag{9}
\]
consisting of the following data:
\begin{enumerate}
\item an integer $t\in\{1,\ldots,T\}$;
\item a $t$-element set $W\subseteq[2t]$ of block indices of the deterministic partition of $[q]$ into $2t$ blocks used in the construction of $A_0$;
\item a color string $\mathbf i=(i_1,\ldots,i_t)\in[k]^t$;
\item an ordered tuple $(u_1,\ldots,u_a)$ of distinct vertices of $G$, with $a$ equal to the total size of the blocks in $W$;
\item an ordered tuple $(v_1,\ldots,v_b)$ of distinct vertices of $G$, with $b=q-a$.
\end{enumerate}
The tuple is \emph{retained} if $u_1,\ldots,u_a,v_1,\ldots,v_b$ are pairwise distinct and, on replaying the process from $A_0:=\{u_1,\ldots,u_a\}$ with the color string $\mathbf i$, the process completes $t$ rounds and its true cap $C:=U(P^t)$ satisfies
\[
 \{u_1,...,u_a,v_1,...,v_b\}\subseteq C,
 \qquad
 |V(G)\setminus C|\ge\alpha tn.
\tag{10}
\]
For a retained description we write
\[
 h:=|V(G)\setminus C|.
\tag{11}
\]
\end{definition}

We now show in Proposition~\ref{prop:sampling} that every colorable sample $S$ gives rise to at least one retained description, and that each retained description determines at most one sample. There are then only few possible descriptions, and each one is realized by a $q$-tuple with exponentially small probability.

\begin{proposition}\label{prop:sampling}
With $T$ and $\alpha$ as above, let
\[
 q:=\max\left\{4T,\left\lceil\frac{4}{\alpha}\log(16k)\right\rceil\right\}.
\tag{12}
\]
For every $\eps$-far graph $G$, a uniformly random $q$-vertex induced subgraph is non-$k$-colorable with probability at least $2/3$. Consequently,
\[
 S_k(\eps)=O\!\left(\frac{k\log k}{\eps}\right).
\tag{13}
\]
\end{proposition}

\begin{proof}
If $n<q$, we sample all vertices. Since $\eps n^2>0$, an $\eps$-far graph cannot itself be $k$-colorable, so the success probability is one.

\textbf{Get a retained description from a $k$-colorable $S$}

Assume $n\ge q$, and equivalently sample a uniformly random ordered tuple $S=(x_1,\ldots,x_q)$ of distinct vertices. Suppose $G[S]$ is $k$-colorable and fix a proper coloring $\varphi:S\to[k]$. Run the trial coloring process with $A=S$ and choose $i_t=\varphi(f_t)$. By Lemma~\ref{lem:invariants}, since $|S|\ge 4T\ge T$, the run completes $T$ rounds and preserves $S$ in every cap. Lemma~\ref{lem:missing} gives some $1\le t\le T$ and $C^*=U(P^t)$ such that
\[
 S\subseteq C^*,
 \qquad
 |V(G)\setminus C^*|\ge\alpha tn.
\tag{14}
\]

The $t$ pivots occupy at most $t$ blocks. Choose a $t$-subset $W\subseteq[2t]$ containing those blocks, adding arbitrary blocks if necessary. Let $A_0$ be the sample vertices in the blocks indexed by $W$, and put $B_0=S\setminus A_0$. Write
\[
 a:=|A_0|,
 \qquad b:=|B_0|=q-a.
\]
Then $\{f_1,\ldots,f_t\}\subseteq A_0\subseteq S$. By Lemma~\ref{lem:stability}, replaying the first $t$ rounds from $A_0$ with the same color string $\mathbf i=(\varphi(f_1),\ldots,\varphi(f_t))$ produces the same cap $C^*$. The ordered $u$-part $(u_1,\ldots,u_a)$ corresponding to $S$ is obtained from $A_0$ by listing its vertices in increasing order of their sample positions (their relative position is the same as in the ordered tuple $S$), and $v_1,\ldots,v_b$ are obtained in the same way from $B_0$. Once $t$ and $W$ are fixed, the correspondence between the ordered $q$-tuples $S$ and $(u_1,...,u_a,v_1,...,v_b)$ is a bijection.

Consequently,
\[
 (t,\,W,\,(\varphi(f_1),\ldots,\varphi(f_t)),\,u_1,\ldots,u_a,\,v_1,\ldots,v_b)
\]
is a retained description in the sense of Definition~\ref{def:description}: the replay from $A_0$ with the color string $\mathbf i$ yields the cap $C^*$, and (14) gives both $S\subseteq C^*$ and $|V(G)\setminus C^*|\ge\alpha tn$. Hence every colorable sample $S$ admits at least one retained description (there may be multiple possible choices of $t$, $W$, $\varphi$). 

\textbf{Recover $S$ from the retained description}

Conversely, the sample $S$ itself is recovered from a retained description as follows: by $W$ and the block partition we know which of the $q$ sample positions belong to the $u$-part and which to the $v$-part, the two tuples are listed in increasing order of their positions, and the vertices of the two tuples are pairwise distinct. Hence a retained description determines at most one ordered sample.

\textbf{Counting the number of retained descriptions.}

The ordered $v$-part of a retained description must consist of distinct vertices of $C\setminus A_0$: in the $v$-positions of the sample we have placed vertices that avoid the true missing set $V(G)\setminus C$, and no vertex may appear in both the $u$-part and the $v$-part. Since $A_0\subseteq C$, there are exactly $n-a-h$ available vertices, and hence at most $P_{n-a-h}^b$ choices for the ordered $v$-part.

For fixed $t,W,\mathbf i$, the block-size estimate gives
\[
 a\le t\left\lceil\frac{q}{2t}\right\rceil\le\frac q2+t,
 \qquad
 b=q-a\ge\frac q2-t\ge\frac q4,
\tag{15}
\]
where $q\ge4T\ge4t$. We first give a common upper bound. For fixed $t,W,\mathbf i$, the numbers $a$ and $b$ are fixed, but the replayed cap and its missing-set size may depend on the ordered $u$-tuple. Let $\mathcal U_{t,W,\mathbf i}$ be the set of ordered tuples $\mathbf u=(u_1,\ldots,u_a)$ of distinct vertices for which replaying the process on $A_0=\{u_1,\ldots,u_a\}$ with color string $\mathbf i$ completes $t$ rounds and yields a cap $C=U(P^t)$ satisfying $A_0\subseteq C$ and $|V(G)\setminus C|\ge\alpha tn$. For $\mathbf u\in\mathcal U_{t,W,\mathbf i}$, write $h(\mathbf u):=|V(G)\setminus C|$ for the missing-set size of its replayed cap.

Conditioned on a particular ordered $u$-tuple, the remaining positions are filled uniformly by an ordered $b$-tuple from $V(G)\setminus A_0$. Thus, writing $h=h(\mathbf u)$ temporarily, the fraction of these completions that avoid the missing set is
\[
\begin{aligned}
\frac{P_{n-a-h}^b}{P_{n-a}^b}
&=\prod_{j=0}^{b-1}\frac{n-a-h-j}{n-a-j}\\
&=\prod_{j=0}^{b-1}\left(1-\frac{h}{n-a-j}\right)\\
&\le\left(1-\frac{h}{n-a}\right)^b\\
&\le\exp\left(-\frac{hb}{n-a}\right)\\
&\le \exp\left(-\frac{hb}{n}\right)\le e^{-\alpha tb}
\le e^{-\alpha tq/4}.
\end{aligned}
\tag{16}
\]
The displayed product is used only when $b\le n-a-h$; otherwise the fraction is zero and the same bound holds. The condition $A_0\subseteq C$ ensures that the entire missing set lies in $V(G)\setminus A_0$. The penultimate inequality uses $h\ge\alpha tn$ and $n-a\le n$.

For fixed $t,W,\mathbf i$, let $N_{t,W,\mathbf i}$ denote the number of retained descriptions. Summing separately over the ordered $u$-tuples, rather than treating $h$ as a constant, gives
\[
\begin{aligned}
\frac{N_{t,W,\mathbf i}}{P_n^q}
&=\frac{1}{P_n^a}\sum_{\mathbf u\in\mathcal U_{t,W,\mathbf i}}
  \frac{P_{n-a-h(\mathbf u)}^b}{P_{n-a}^b}\\
&\le\frac{|\mathcal U_{t,W,\mathbf i}|}{P_n^a}
  e^{-\alpha tq/4}
\le e^{-\alpha tq/4}.
\end{aligned}
\]
Here we used $P_n^q=P_n^aP_{n-a}^b$ and $|\mathcal U_{t,W,\mathbf i}|\le P_n^a$. 

There are at most $\binom{2t}{t}\le4^t$ choices for $W$ and $k^t$ choices for the color string. Every colorable sample yields at least one retained description, even though its proper coloring may depend on the sample. Summing the preceding bound over all strings and block choices therefore gives, by (12),
\[
\begin{aligned}
\Pr[G[S]\text{ is }k\text{-colorable}]
&\le\sum_{t=1}^T(4k)^t e^{-\alpha tq/4}\\
&\le\sum_{t=1}^T(4k)^t(16k)^{-t}\\
&=\sum_{t=1}^T4^{-t}\le\frac13.
\end{aligned}
\]
Thus the induced subgraph is non-$k$-colorable with probability at least $2/3$. Finally, $T=O(k/\eps)$, $\frac{1}{\alpha}=\frac{4(k+1)}{\eps}$, and $\log(16k)=O(\log k)$, proving (13).\end{proof}

\section{Concluding remarks}

The proof combines the container method with a palette-based coloring process tailored to $k$-colorability. 

The reason we can improve upon the bound of \(\cite{BS2023,KushnirShapira2026}\) is that our method allows us to introduce a structure to choose \(t\) blocks out of \(2t\) blocks. By contrast, the fingerprint in \(\cite{BS2023,KushnirShapira2026}\) is effectively enumerated by a factor of the form \(\binom{q}{t}\), together with the corresponding color assignments. Without nontrivial modifications, their mechanism does not support a \(\binom{2t}{t}\)-type encoding like the one used in our algorithm. In the Kushnir--Shapira construction, the container at round \(t\) does not depend on the entire independent set \(I\); it depends only on the fingerprint from the first \(t\) rounds. However, every vertex in \(\operatorname{vert}(F_t)\) must be assigned a color in order to determine \(F_t\). If one instead starts with a subset \(V\subseteq \operatorname{vert}(I)\) without specifying the colors of all vertices in \(V\), different choices of colors can lead to different processes and hence to different projected containers \(\operatorname{vert}(C_t)\). See Appendix~\hyperref[app:worked-example]{A} for a concrete example illustrating the algorithms.

A second important ingredient is how the assumption that the original graph is $\eps$-far from $k$-colorable forces some properties. Intuitively, call a vertex ''bad'' at a given stage if the process is unwilling to assign it a color. If the sampled subgraph is $k$-colorable, the process can preserve its proper coloring and continue; conversely, if too few vertices become bad throughout the process, the reconstruction lemma produces a coloring of the original graph with fewer than $\eps n^2$ monochromatic edges, contradicting $\eps$-farness. Moreover, a coloring process in which many such "bad" vertices appear occurs with very low probability if the subgraph $S$ is $k$-colorable. 

\section*{Statement on AI usage}

The author had early discussions with ChatGPT about how to approach the palette system. The author then conceived how palettes could be used within the proof and wrote out the complete argument. ChatGPT was subsequently used to edit the author's first draft of the proof into a form conforming to standard conventions of mathematical exposition. After this, the author reviewed the draft in full and polished the exposition to obtain the final version of this paper. The author takes full responsibility for all content in this paper.

\section*{Appendix: A Worked Example of the Fingerprint--Container and Palette Algorithms}\phantomsection\label{app:worked-example}

\begin{quote}
\textbf{Fingerprint and Container Generation} (~\cite{BS2023,KushnirShapira2026}).\\
\textbf{Input:} A graph $H=(V,E)$, an independent set $I\subseteq V$, and a fixed vertex order for tie-breaking.\\
\textbf{Output:} Fingerprints $F_1,\ldots,F_{|I|}$ and containers $C_1,\ldots,C_{|I|}$.
\begin{tabbing}
\hspace{1.5em}\=\hspace{1.5em}\=\kill
1. \> $F_0\gets\varnothing$; $C_0\gets V$.\\
2. \> \textbf{for} $t=1,2,\ldots,|I|$ \textbf{do}\\
3. \> \> Choose $f_t\in I\setminus F_{t-1}$ with maximum degree in $H[C_{t-1}]$, breaking ties by the fixed order.\\
4. \> \> $X_t\gets\{w\in C_{t-1}:\deg_{H[C_{t-1}]}(w)>\deg_{H[C_{t-1}]}(f_t)\}$.\\
5. \> \> $Y_t\gets N_H(f_t)\cap C_{t-1}$.\\
6. \> \> $F_t\gets F_{t-1}\cup\{f_t\}$.\\
7. \> \> $C_t\gets C_{t-1}\setminus(X_t\cup Y_t)$.\\
8. \> \textbf{end for}\\
9. \> \textbf{return} $(F_1,C_1),\ldots,(F_{|I|},C_{|I|})$.
\end{tabbing}
\end{quote}

We give a concrete example illustrating both the fingerprint--container algorithm of~\cite{BS2023,KushnirShapira2026} and our palette algorithm. The example is a 3-coloring instance on a graph $G$. Here $k=3$: the auxiliary graph is the triplicate construction $H=K_3*G$, where each vertex $u$ of the original bipartite graph $G$ is replicated into three copies $u_1,u_2,u_3$ forming a triangle. The original graph $G$ has vertex set $\{a,b,c,v,x,y\}\cup\{w,r,s,t\}$ and edge set
\[
 E(G)=\{vw,\,vr,\,vs,\,vt,\,yr,\,ys,\,xr,\,wa,\,wb,\,wc\}.
\]

\textbf{Figure~\ref{fig:bipartite-base}} and \textbf{Figure~\ref{fig:triplicate}} show the base graph $G$ together with the resulting graph $H$ obtained by the triplicate construction. 

We first illustrate the first two rounds of our palette process on this graph. Let $A=\{v,x,y\}$, and fix the total order so that $x$ precedes $y$. Initially every palette is $[3]$, so the first pivot is $f_1=v$, whose degree in $G$ is larger than those of $x$ and $y$. Choose any prescribed color $i_1\in[3]$ for $v$. In the update from $P^1$ to $Q^1$, color $i_1$ is deleted from the palettes of $w,r,s,t$. At the start of round $2$, both $x$ and $y$ have minimum palette-degree $0$ (consider color $i_1$). Thus the tie-breaking order selects $f_2=x$ and $D_2=0$. In forming $P^2$, the vertices $w,r,s$ have positive minimum palette-degree (respectively $3,2,1$), so their palettes are emptied. In particular, $P^2(x)=[3]$, so the second prescribed color $i_2$ may be any of $1,2,3$.

\begin{figure}[htbp]
\centering
\includegraphics[width=0.55\textwidth]{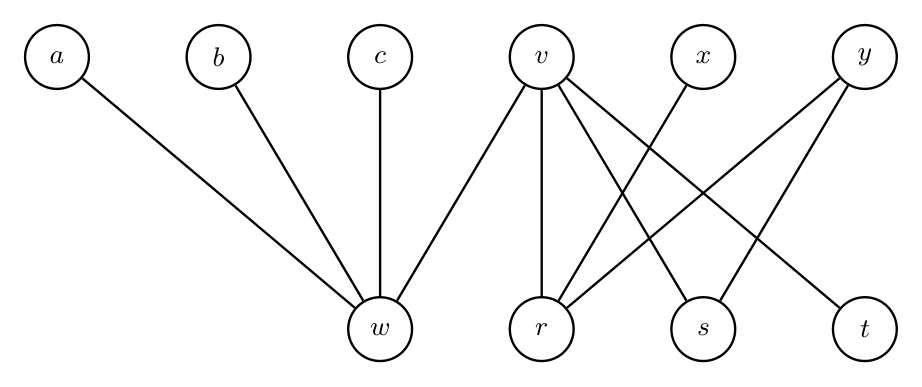}
\caption{The original bipartite graph $G$.}
\label{fig:bipartite-base}
\end{figure}

\begin{figure}[htbp]
\centering
\includegraphics[width=0.8\textwidth]{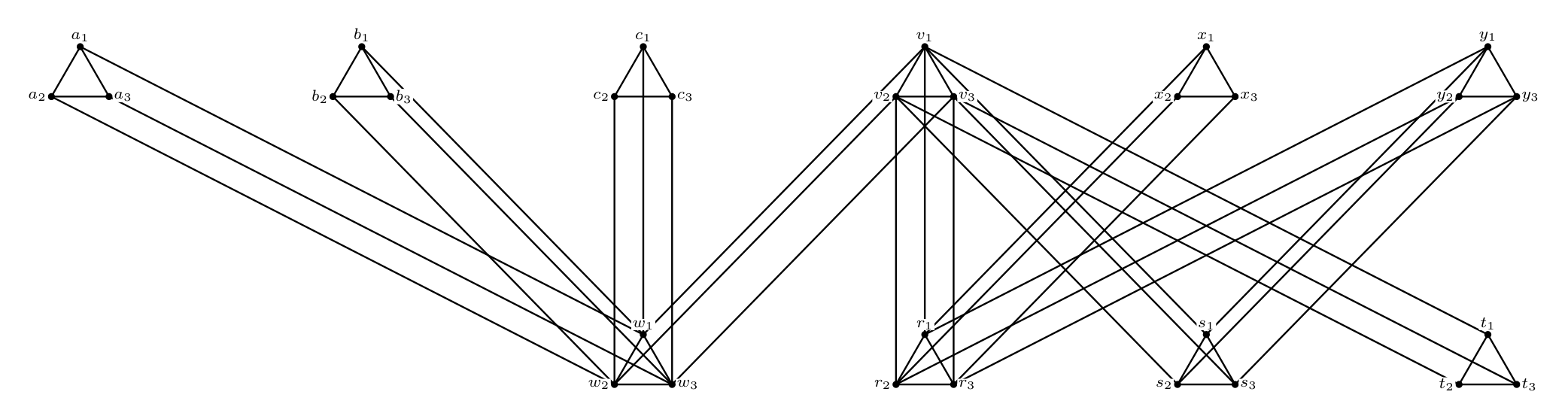}
\caption{The triplicate graph $H=K_3*G$. Each vertex of $G$ is replaced by a triangle of three copies.}
\label{fig:triplicate}
\end{figure}

Then we show a simulation of the algorithm in ~\cite{BS2023,KushnirShapira2026}. We consider two possible $I$, and run the algorithm for $t=1,2,3$. $I_1=\{v_1,x_2,y_1\}$, $I_2=\{v_1,x_1,y_2\}$ are inputs. In the graph $H[C_0]=H$, $\text{deg}(v_1)=6$, $\text{deg}(x_1)=\text{deg}(x_2)=3$, $\text{deg}(y_1)=\text{deg}(y_2)=4$.

For both cases, the first iteration the algorithm selects the vertex of $I$ with largest degree in the current graph; this is $v_1$. \textbf{Figure~\ref{fig:ver2}} illustrates the first iteration of the container algorithm.  The figure shows the resulting container $C_1$: all copies of the neighbors of $v_1$ (namely $v_2,v_3,w_1,r_1,s_1,t_1$) together with their incident edges have been deleted.

\begin{figure}[htbp]
\centering
\includegraphics[width=0.8\textwidth]{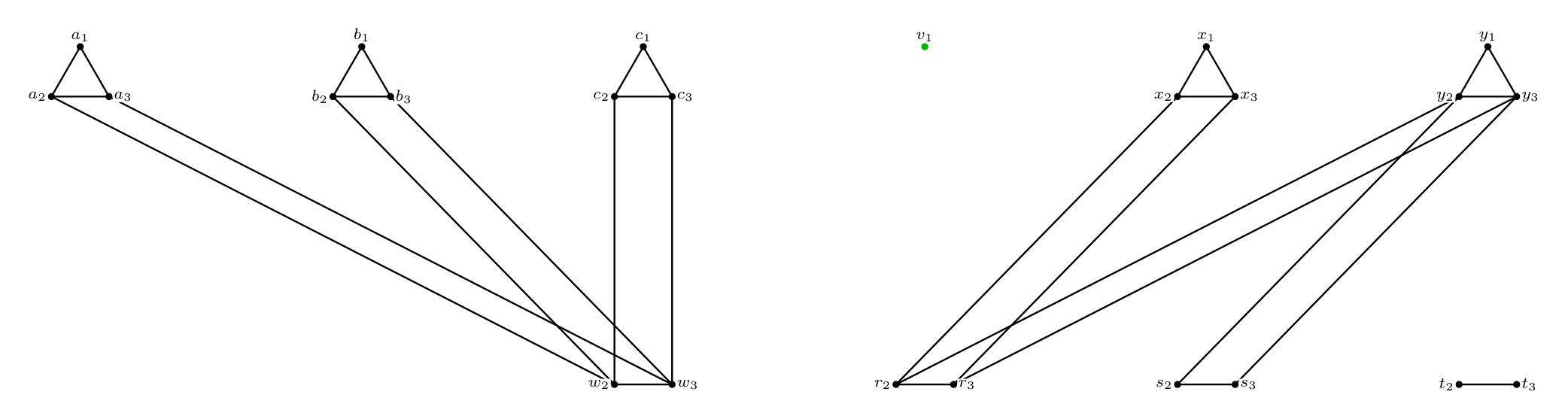}
\caption{The container $C_1$ after the first iteration of the Blais--Seth fingerprint--container generator on $H=K_3*G$, with $v_1$ selected from $I=\{v_1,x_2,y_1\}$. The green vertex is $v_1$; its neighbors $v_2,v_3,w_1,r_1,s_1,t_1$ and all incident edges have been removed.}
\label{fig:ver2}
\end{figure}

Case 1: $I_1$. $\text{deg}_{H(C_1)}(x_2)=3, \text{deg}_{H(C_1)}(y_1)=2$. 
$t=2$ step: The algorithm chooses $f_2=x_2$. $X\cup Y=\{x_1,x_3,r_2,w_2,w_3,y_2,y_3\}$. $\text{vert}(C_2)=\{a,b,c,v,x,y,r,s,t\}$. 
$t=3$ step: The algorithm chooses $f_3=y_1$. $C_3=\{v_1,x_2,y_1,r_3\}$. $\text{vert}(C_3)=\{v,x,y,r\}$. 

Case 2: $I_2$. $\text{deg}_{H(C_1)}(x_1)=2, \text{deg}_{H(C_1)}(y_2)=4$. 
$t=2$ step: The algorithm chooses $f_2=y_2$. $X\cup Y=\{y_1,y_3,r_2,s_2\}$. $\text{vert}(C_2)=\{w,a,b,c,v,x,y,r,s,t\}$.
$t=3$ step: The algorithm chooses $f_3=x_1$. $C_3=\{v_1,y_2,x_1,a_1,b_1,c_1,r_3,s_3,t_2,t_3\}$. $\text{vert}(C_3)=\{v,x,y,a,b,c,r,s,t\}$. 

In conclusion, this example has $\text{vert}(I_1)=\text{vert}(I_2)$, and in the first, second, and third steps of the algorithm, the chosen colors are 1, 2, and 1, respectively. For each $t\in\{2,3\}$, the projected containers $\text{vert}(C_t)$ obtained from inputs $I_1$ and $I_2$ are different. This example shows that the container process depends not only on $\text{vert}(I)$, but also on the colors assigned to its vertices; moreover, even specifying the colors selected in the first few iterations does not determine a unique projected container $\text{vert}(C_t)$, since the two processes have the same selected-color sequence but different pivot vertices and different projected containers.

\end{document}